\documentclass[american,aps,prl,reprint, superscriptaddress]{revtex4-1}
\usepackage{amsmath,amscd,amsthm,amssymb}
\usepackage{bbm}
\usepackage{mathrsfs}
\usepackage{graphicx}
\usepackage{epsf,epstopdf}
\usepackage{caption3}
\usepackage[all]{xy}
\usepackage{tikz}

\usepackage[unicode=true,pdfusetitle, bookmarks=true,bookmarksnumbered=false,bookmarksopen=false, breaklinks=false,pdfborder={0 0 0},backref=false,colorlinks=false] {hyperref}
\hypersetup{ colorlinks,linkcolor=myurlcolor,citecolor=myurlcolor,urlcolor=myurlcolor}
\usepackage{graphics,epstopdf,graphicx, amsthm, amsmath, amssymb, times, braket, color, bm}
\usepackage[up]{subfigure}
\usepackage{cleveref}

\usepackage{marvosym} 

\definecolor{myurlcolor}{rgb}{0,0,0.7}

\def\be{\begin{equation}}
\def\ee{\end{equation}}
\def\bea{\begin{eqnarray*}}
\def\eea{\end{eqnarray*}}

\theoremstyle{plain}
\newtheorem{thrm}{\protect\theoremname}

\providecommand{\theoremname}{Theorem}

\newcommand*{\myproofname}{Proof}

\makeatother
\newcommand{\DOI}[2]{\href{https://doi.org/#1}{#2}}

\theoremstyle{definition}

\theoremstyle{remark}

\graphicspath{{figs/}}

\begin{document}

 \author{Y.T. Tsui}
 \email{YunTao\_Cui@outlook.com}
 \affiliation{School of Mathematical Sciences, Harbin Engineering University, Harbin 150001, People's Republic of China}

 \author{Sunho Kim}
 \email{Corresponding author: kimsunho81@hrbeu.edu.cn}
 \affiliation{School of Mathematical Sciences, Harbin Engineering University, Harbin 150001, People's Republic of China}

\title{Generalized wave-particle-Mixdness triality for n-path interferometers}
\begin{abstract}
The wave-particle duality, as one of the expressions of Bohr complementarity, is usually quantified by path predictability and the visibility of interference fringes. With the development of quantum resource theory, the quantitative analysis of wave-particle duality is increasing, most of which are expressed in the form of specific functions. In this paper, we obtain the path information measure for pure states by converting the coherence measure for pure state into a symmetric concave function. Then we prove the function as a path information measure is also valid for mixed states. Furthermore, we also establish a generalized wave-particle-mixedness traility. Although the mixedness proposed in the text is not a complete mixedness measure, it also satisfies some conditions of mixdness measure. From the perspective of resource theory, the path information we establish can be used as the measure of the resource of predictability, and the triaility relationship we establish reveals the relationship among coherence, predictability, purity and mixdness degree to a certain extent. According to our method, given either coherence measure or path information, a particular form of wave-particle-mixedness traility can be established. This will play an important role in establishing connections between wave, particle and other physical quantifiers.
\end{abstract}
\maketitle

\section{Introduction}

Wave-particle duality is a fundamental property of quantum, which plays a pivotal role in quantum mechanics and is ubiquitous in the microcosmic world. The earliest research on wave-particle duality can be traced back to the complementarity principle proposed by Bohr in 1982 \cite{ref1}, but it's just a qualitative idea. Wootters and Zurek were the first ones to quantitatively analyze wave-particle duality \cite{ref2}. Later, Greenberger and Yasin proposed a complementarity relation, they analyzed the problem with the assumption the unequal beams in a two-path interferometer allows for that unequal beams in a two-path interferometer allows for predicting, to a certain degree, which of the two paths the quantum followed. They established the duality relation \cite{ref3}:
\begin{equation}\label{eq1}
P^2+W^2\leq1
\end{equation}
Among them, $P$ quantifying the particle feature (path information) and $W$ quantifying the wave feature (interference) of the state in an interferometric setup \cite{ref4}.

As the theory progressed, Englert derived a new duality relation by including a path-detector into the interferometer which used to determine which path a specific quantum took \cite{ref5}:
\begin{equation}\label{eq2}
D^2+V^2\leq1
\end{equation}
In this inequality, $D$ is the distinguishability of the possible detector states and $V$ is a measure of the quality of the interference fridge. The equality holds if the quality of the interference fridge. Investigations on quantitative measures of wave and particle properties in a multi-path interferometer were initiated by Durr in \cite{ref6}, where criteria for generalized predictability and generalized visibility were brought up and later further pursued in \cite{ref7,ref8,ref9}.

In recent years, with the establishment of quantum resource theory, the quantitative analysis of wave-particle duality has returned to the public's view. These studies about the wave-particle duality have revealed close relationships and interplays between predictability, distinguishability, visibility \cite{ref6,ref7,ref8,ref10,ref11,ref12}, and some quantum informational concepts such as asymmetry \cite{ref13}, entropy \cite{ref14,ref15,ref16}, entanglement \cite{ref17,ref18,ref19,ref20,ref21,ref22}, coherence \cite{ref23,ref24,ref25,ref26,ref27,ref28,ref29,ref30,ref31,ref32,ref33,ref34,ref35}, and purity \cite{ref36}.
In addition, research in this area has been increasingly establishing a clearer complementary relationship until recently. In the pure case, the wave-particle duality is a strictly complementary relation, but attempts to generalize the equation to mixed states always turn it into an inequality. Through further analysis of the n-path interferometer with the addition of the path detector, Roy proposed a coherence-path predictability-I concurrence triality \cite{ref37}, I concurrence is a normalized entanglement measure with some functional relationship to the generalized concurrence.
Later, new triality relation was proposed By splitting the path distinguishability into path predictability and entanglement measure \cite{ref38}, this new triality formalizes entanglement as the third quantity.
In contrast, a very concise triality relationship was inferred through Fu et al., by combining quantum uncertainty with coherence and path information without adopting interferometric analysis methods \cite{ref39}.

In this paper, we analyze whether all coherence measures satisfy the neat complementarity relation following the dervation of \cite{ref39}. We transform the coherence measure of pure states into a symmetric concave function related only to main diagonal elements, and we find that all coherence measures can establish a neat complementarity relation, the function which is complementary to the symmetric concave function converted from the coherence measure perfectly meets the requirements of the particle property measure. Moreover, we can also obtain a triality relation in the mixed state, although the third quantity is not a complete measure of mixdness, it still satisfies some properties of mixdness measure.

\section{PRELIMINARIES}

In this part we will introduce some basic knowledgement we're going to use. In subsection \uppercase\expandafter{\romannumeral1}, We will briefly introduce the quantum coherence resource thory and the conditions for coherence measures. It is also shown that the coherence measure in any pure state can be expressed as a symmetric concave function, and the coherence measure in mixed state can be constructed by convex tops. And in subsection \uppercase\expandafter{\romannumeral2}, we will enumerate the basic knowledge of wave-particle duality and its recent development, it includes a brief introduction to the theoretical analysis of multipath interferometer  experiments and the recently established wave-particle duality. At last, we will enumerate the conditions that must be satisfied for the measure of wave and particle properties, respectively in subsection \uppercase\expandafter{\romannumeral3}. We wondered whether the measures of particle and wave we propose next meet these conditions.

\subsection{\uppercase\expandafter{\romannumeral1}. Quantum Coherence}

Coherence is a fundamental feature of quantum physics, which represents possible superposition between orthogonal quantum states. Similarly, it is widely believed that quantum superposition is a manifestation of the fluctuating nature of quantum particles. Therefore, the quantum coherence has a strong correspondence with the wave properties of quantum particles.

Based on Baumgratz et al.’s suggestion \cite{ref40}, any proper measure of coherence $C$ must satisfy the following axiomatic postulates.

\hspace*{\fill}

$(C1)$ The coherence measure vanishes on the set of incoherent states, $C(\rho)=0$ for all $\rho\in I$;

$(C2a)$ Monotonicity under incoherent operation $\Phi$, $C(\Phi(\rho))\leq C(\rho)$;

$(C2b)$ Monotonicity under selective measurements on average, $\sum_n p_n C(\rho_n)\leq C(\rho)$, where $p_n=tr(K_n\rho K^\dagger _n)$, $\rho=\frac{1}{p_n}K_n\rho K^\dagger_n$, for all $\{K_n\}$ with $\sum_n K^\dagger_n K_n =I$ and
\begin{equation}\label{7}
K_n\rho K^\dagger_n / Tr(K_n\rho K^\dagger_n) \in I
\end{equation}
for all $\rho\in I$;

$(C3)$ Non-increasing under mixing of quantum state (convexity),
\begin{equation}\label{8}
C(\sum_n p_n \rho_n)\leq \sum_n p_n C(\rho_n)
\end{equation}
for any ensemble $\{p_n,\rho_n\}$.

\hspace*{\fill}

In general, coherence monotone must satisfies the conditions $C1$ $C3$ and $C2a$, and coherence measure must meet all the conditions.

Corresponding to the functional form of the entanglement measure, A proposes the functional form of the coherence measure was given by Du et al. in \cite{ref41} as follow.

Let $\Omega=\{\bm{x}=(x_1,x_2,\dots,x_d)^t\mid\sum_{i=1}^{d} x_i=1\quad\text{and}\quad x_i\geq0\}$, here $(x_1, x_2,\dots,x_d)^t$ denotes the transpose of row vector $(x_1,x_2,\dots,x_d)$. And let $\pi$ be an arbitrary permutation of $\{1,2,\dots,d\}$, $P_\pi$ be the permutation matrix corresponding to $\pi$ which is obtained by permuting the rows of a $d\times d$ identity matrix according to $\pi$. Given any nonnegative function $f:\Omega\mapsto R^+$ such that it is

\hspace*{\fill}

1. If only one of the terms is 1, and the rest is 0, the function is zero, i.e.
\begin{equation}\label{9}
f(P_\pi(1,0,\dots,0)^t)=0
\end{equation}
for every premutation $\pi$,

2. Invariant under any permutation transformation $P_\pi$, i.e.
\begin{equation}\label{10}
f(P_\pi \bm{x})=f(\bm{x})
\end{equation}
for every $\bm x\in\Omega$,

3.  concave, i.e.
\begin{equation}\label{11}
f(\lambda\bm{x}+(1-\lambda)\bm{y})\geq\lambda f(\bm x)+(1-\lambda)f(\bm{y})
\end{equation}
for any $\lambda\in[0,1]\quad\text{and}\quad\bm{x},\bm{y}\in\Omega$.

Then a coherence measure can be derived by defining it for  pure states (normalized vectors $\ket{\psi}=(\psi_1,\psi_2,\dots,\psi_d)^t$ in the fixed basis$\{\ket{i}\}^d_{i=1}$) as
\begin{equation}\label{12}
C_f(\ket{\psi}\bra{\psi})=f((|\psi_1|^2,|\psi_2|^2,\dots,|\psi_d|^2)^t)
\end{equation}
The function $f$ is extended to the whole set of density matrices by the form of convex roof structure.
\begin{equation}\label{13}
C_f(\rho)=\min_{p_j,\rho_j}\sum_jp_jC_f(\rho_j)
\end{equation}
where the minimization is to be performed over all the pure-state ensembles of $\rho$, i.e., $\rho=\sum_jp_j\rho_j$.

\hspace*{\fill}

For any pure state $\ket{\psi}$, all symmetric concave functions $f$ satisfying the above conditions are coherence measures. In reverse, a symmetric concave function can be found for any coherence measure. Moreover, the above conditions also point out a way to construct a coherence measure  of mixed states from a coherence measure of pure states. There are many instances, such as $\alpha$-entropy\cite{ref42}, flidelity coherence measure\cite{ref43}, and so on.

\subsection{\uppercase\expandafter{\romannumeral2}. Wave-Paricle Duality}

There are two main types of n-path interferometers, with and without path dectors, see (FIG.1). A particle is fired from emitter, travels through one of n-paths as the manifestation of the property of particle, and ends up hitting the screen to create a number of interferenve fringes as the manifestation of the property of wave. With the development of coherence resource theory, interference fringes on the screen can be formulated as coherence, and which path the particle passes depends on the path information.

\begin{figure}[h]
    \centering
    \includegraphics[scale=0.4]{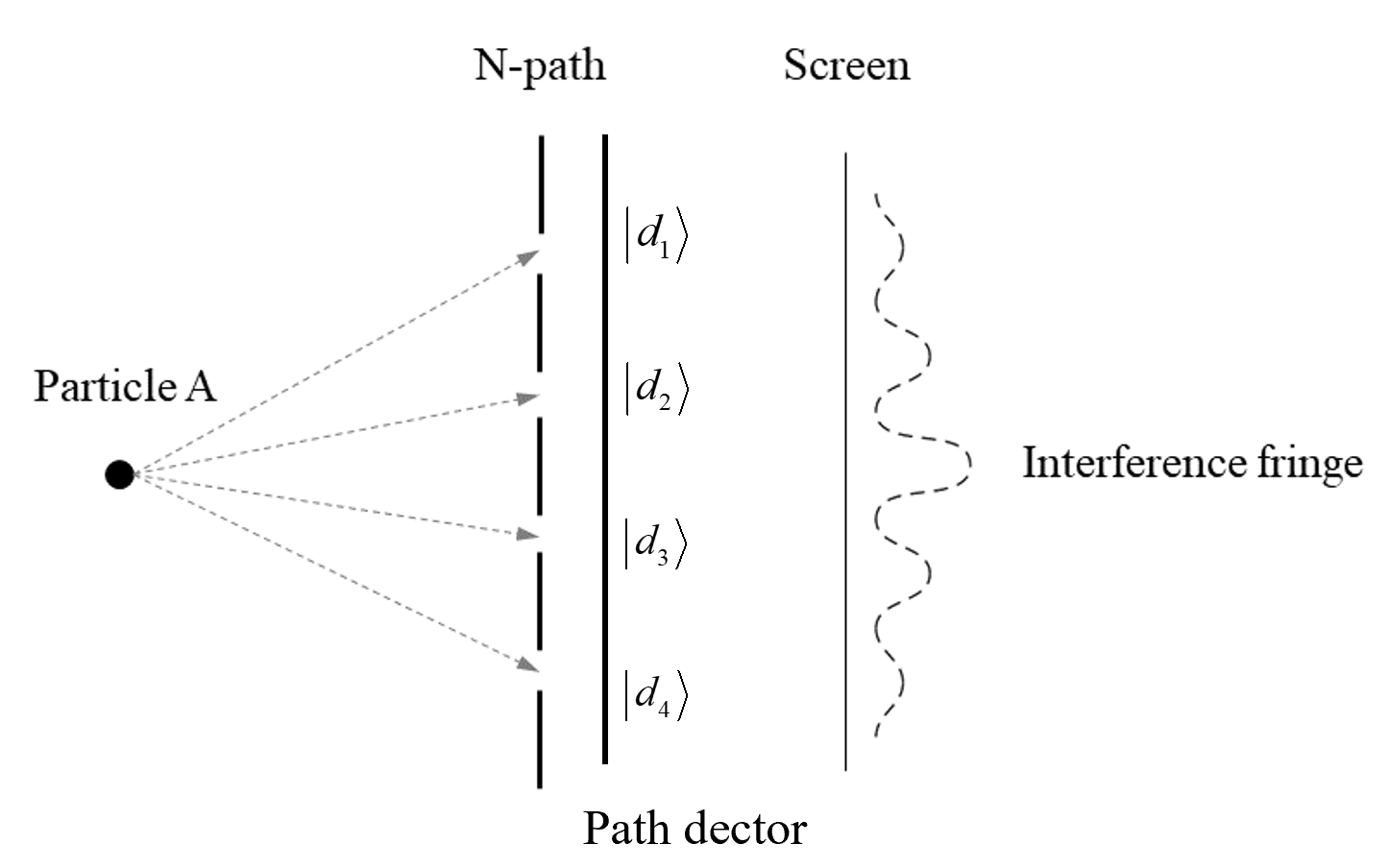}
    \caption{Diagram of an n-path interferometer with path detector, a detector device is installed on each path, and when a particle passes through the i path, the detector changes from $\ket{d_0}$ to $\ket{d_i}$.}
    \label{FIG.1.}
\end{figure}

For the interferometer without path dectors, we cannot clearly determine which path the particle takes, so we can only predict the path with the greatest probability through the properties of the particle itself, which is called the path predictability of the particle. This in itself is a guesswork, so there are many ways to characterize path predictable rows, such as the classic single-bet protocol \cite{ref4} and the multi-bet protocol \cite{ref6,ref7}.

And for the interferometer with path detectors, we can clearly determine which path the particle takes, so we do not have to guess but judge the path, which is called the path distinguishability of the particle. The initial detector state is $\ket{d_0}$, when the particle $\rho=\sum_{ij}\rho_{ij}\ket{\psi_i}\bra{\psi_j}$ pass through the ith path, become entangled with the detector, through the controlled unitary interaction, each of the paths is marked with a detector state $\ket{d_i}$, the path distinguishability is equivalent to discriminating the detector states. In other words, if the quanton passes through the ith path, the resulting the detector states becomes $\ket{d_i}$ with the probability $|\rho_{ii}|$. Beacuse dector states ${\ket{d_i}}$ are not mutually orthogonal in general, so we have partial knowledge about which-path information.

Although the above path predictability and path distinguishability belong to a kind of path information, there are subtle differences between them. For the case of no path detector, the particle does not change when it passes through the path and finally hits the screen after being emitted from the emitter, still remains as
\bea
\rho=\sum_{ij}\rho_{ij}\ket{\psi_i}\bra{\psi_j}
\eea
but in the case of path detectors, the emitter fire a state $\rho=\sum_{ij}\rho_{ij}\ket{\psi_i}\bra{\psi_j}$, when passes through the detector, it becomes entangled with the detector and becomes
\bea
\rho_{sd}=\sum_{i}^{n}\sum_{j}^{n}\rho_{ij}\ket{\psi_i}\bra{\psi_j}\otimes\ket{d_i}\bra{d_j}
\eea
by tracing out over the path-detector states in $\rho_{sd}$, one gets the reduced density matrix for the quanton
\bea
\rho_{s}=\sum_{i}^{n}\sum_{j}^{n}\rho_{ij}\ket{\psi_i}\bra{\psi_j}\bra{d_i}\ket{d_j}
\eea
this will be the state that finally hits the screen.

The first case involves only operations on a single system, while the second case involves two systems (state and path detectors). The path information is different in different cases (the path distinguishability are used in the case with detectors, and the path predictability is used in the case without detectors), and the corresponding coherence measure is also different($\rho_{s}$ is measured when there are detectors, and $\rho$ is measured when there are no detectors). In both cases, the path information and coherence are complementary, which also satisfies the Theorem 1 and Theorem 2 we proposed later.

\subsection{\uppercase\expandafter{\romannumeral3}. Wave-Paricle Mixedness Triality}

Let's review some of the requirements for wave-particle duality in general \cite{ref6,ref7,ref8,ref9}.

Following \cite{ref6,ref7,ref8}, The quantification of the quantum wave feature (we denote it as $W(\rho)$) should satisfies the following reasonable requirements:

\hspace*{\fill}

$(1)$ $W(\rho)$ reaches its global minimum if the state $\rho$ is classical (i.e. diagonal in the computational basis).

$(2)$ $W(\rho)$ reaches its global maximum if the state $\rho$ is pure and a uniform superposition of the states in the computational basis (i.e. $\bra{i}\rho\ket{i}=1/n$ for all $i$).

$(3)$ $W(\rho)$ is invariant under permutations of the diagonal elements $\bra{i}\rho\ket{i}$ of $\rho$.

$(4)$ $W(\rho)$ is convex.

\hspace*{\fill}

In a dual fashion, a quantification of the quantum particle feature (we denote it as $P(\rho)$) should have the following properties \cite{ref6,ref8}:

\hspace*{\fill}

$(1)$ $P(\rho)$ reaches its global maximum if $\bra{i}\rho\ket{i}=1$ for some $i$.

$(2)$ $P(\rho)$ reaches its global minimum if $\bra{i}\rho\ket{i}=1/n$ for all $i$.

$(3)$ $P(\rho)$ is invariant under permutations of the diagonal elements $\bra{i}\rho\ket{i}$ of $\rho$.

$(4)$ $P(\rho)$ is convex.

\hspace*{\fill}

All the above mathematical requirements are motivated by intuitive physical considerations \cite{ref6,ref7,ref8}, and all the wave properties and particle properties should satisfy the above conditions.

Next we introduce a very important triality relation, and part of our results can be seen as a general generalization of this relation.

In \cite{ref39}, it is pointed out that coherence can be measured by the uncertainty of states \cite{ref44}, and the path traversed by particles can be quantified by the path certainty, that is, the certainty of the measurement, in this way, a new wave-particle-mixedness is proposed as follow. This, to some extent, has inspired our study.
\begin{equation}\label{14}
P(\rho|\Pi)+W(\rho|\Pi)+M(\rho)=1
\end{equation}
Where, $P(\rho|\Pi)=\Sigma_{i=1}^n\bra{i}\rho\ket{i}^2$ is the measurement certainty, $W(\rho|\Pi)=\Sigma_{i\not=j}|\bra{i}\rho\ket{i}|^2$ represents the state uncertainty, and $M(\rho)=1-tr\rho^2$ represents the mixedness of particle, which is the uncertainty possessed by the particle itself, and can be associated with other physical quantifiers under certain conditions.

\section{New wave-particle relation}

Next we present the main results of this article. In subsection \uppercase\expandafter{\romannumeral1}, we start our derivation from the basic case of pure states, and obtain a generalized neat wave-particle complementary relation for pure states. In subsection \uppercase\expandafter{\romannumeral2}, we try to extend the complementary relation obtained in \uppercase\expandafter{\romannumeral1} to the mixed states, and we get a generalized wave-particle complementary relation finally.

\subsection{\uppercase\expandafter{\romannumeral1}. Generalized Wave-Particle relation for pure states}

Reviewing the relations proposed in the previous literature, all of them are in a specific functional form, their special form ensures the complementarity of wave and particle, and they can be related to other physical quantifiers (relative entropy and mutual information \cite{ref24}, measurement uncertainty and quantum uncertainty \cite{ref39}, etc). However, not all coherence measures can find a perfect corresponding function to form the duality relation, but these coherence measures have good correlations with other physical quantifiers (such as fidelity \cite{ref43}, etc). So we wonder whether there is a relation that can be applied to all forms of coherence measures.

Based on \cite{ref41}, all coherence measures $C_f(\ket{\psi})$ for pure state $\ket{\psi}=(\psi_1,\psi_2,\dots,\psi_n)^t$ can be transformed as functions with respect to the main diagonal elements of matrix
\bea
C_f(\ket{\psi})=f(|\psi_1|^2,|\psi_2|^2,\dots,|\psi_n|^2)
\eea
It is easy to verify that all functional forms $f$ satisfy the conditions regarding the quantification of the wave feature.

Whereas in view of the wave-particle duality relations that have been proposed can be neat complementary in the pure state, so, out of intuition, we define
\bea
D_f(\ket{\psi})=1-C_f(\ket{\psi})=1-f(|\psi_1|^2,|\psi_2|^2,\dots,|\psi_n|^2)
\eea
as a measure of the properties of the corresponding particle, as well as, path information.

Next we will show that the function $D_f(\ket{\psi})$ we have defined is perfectly consistent with the requirements of particle feature quantification.

\begin{thrm}
	For any given coherence measure $C_f(\ket{\psi})$ for pure state $\ket{\psi}$, there will always be a corresponding path information $D_f(\ket{\psi})$, satisfy $C_f(\ket{\psi})+D_f(\ket{\psi})=1$.
\end{thrm}
\begin{proof}
	As we defined it before:
	\bea
	C_f(\ket{\psi})=f(|\psi_1|^2,|\psi_2|^2,\dots,|\psi_n|^2)
	\eea
	\bea
	D_f(\ket{\psi})=1-f(|\psi_1|^2,|\psi_2|^2,\dots,|\psi_n|^2)
	\eea
	
	\hspace*{\fill}
	
	$(1)$ Set $\bra{i}\ket{\psi}\bra{\psi}\ket{i}=1$ for a definite i, then $\ket{\psi}$ should be of the following form
	\bea
	\ket{\psi}=(0_1,\dots,1_i,\dots,0_n)^t
	\eea
	except for the ith term which is 1, the rest of the term is 0. Base on  (9), and
	\bea
	E_{1i}(1,0,\dots,0)=\ket{\psi}
	\eea
	so we gain $f(\ket{\psi})=0$, meanwhile $D_f(\ket{\psi})$ reaches its global maximum 1.
	
	$(2)$ Set $\bra{i}\ket{\psi}\bra{\psi}\ket{i}=1/n$ for all $i$, then $\ket{\psi}=\sum_{i=1}^{n}\frac{1}{\sqrt{n}}\ket{i}$ is a maximally coherence state. At this time, the coherence of state $\ket{\psi}$ reaches its maximum 1 and the complementary path information to it reaches its globle minimum 0, $D_f(\psi)=0$.
	
	$(3)$ That's straight from (10), $D_f(\ket{\psi})$ is invariant under permutations of the diagonal elements of $\ket{\psi}\bra{\psi}$.

	$(4)$ It is easy to see that $D_f=1-f$ is convex because of $f$ is concave.
	
	\hspace*{\fill}
	
\end{proof}

	In summary, $D_f(\ket{\psi})$ what we have defined satisies the necessary conditions for the quantification of particle feature.

\subsection{\uppercase\expandafter{\romannumeral2}. A Generalized wave-paritcle relation}

Now we extend the results we obtained to mixed states. Unlike in the pure state, the complementarity relations in the mixed state are mostly unneat, a third term is needed to complement it, and this third term is not the same for different analysis methods. In the interferometer considering path detectors, this third term is formulated as a measure of entanglement \cite{ref39}, while in the interferometer without path detectors, this third term is formulated as the mixdness of the state \cite{ref33,ref38,ref45}. This paper does not consider the case of including detectors, and in fact the analysis of the case including detector is also pretty hard.

For some functions $f$, the same form satisfies the conditions for the coherence measure for pure states but not for mixed states. At this point we can construct a coherence measure suitable for mixed states by convex roof structure of $f(\ket{\psi})$.
\bea
C_f(\rho)=\min_{p_j,\rho_j}\sum_jp_jC_f(\rho_j)
\eea
$\rho=\sum_jp_j\rho_j$ is the minimization of all the pure state ensembles of $\rho$.

Next, we analyze the path information, From \cite{ref39,ref15,ref16}, we know that the path information is only related to the main diagonal elements of the state matrix, and pure and mixed states with the same main diagonal elements have the same path information.

For any given mixed state, decoherence is first performed on this state,
\bea
\widetilde{\rho}=\Delta(\rho)
\eea
since coherence and path information are two complementary physical quantities, so decocoherence operation can only clear the coherence of states without worrying about the change of path information. Next we complement some off-diagonal elements to the state after decoherence $\widetilde{\rho}$. These off-diagonal elements are not chosen arbitrarily, we set the element in row $i$ and column $j$ of the matrix to be the root of the product of the $i$th and $j$th entries of the main diagonal.
\bea
\widehat{\rho}_{ij}=\sqrt{\widetilde{\rho}_{ii}\widetilde{\rho}_{jj}}
\eea
Thus we have changed a mixed state $\rho$ into a pure state $\ket{\phi}$ by the above operation. The corresponding pure state have the following form:
\bea
\ket{\phi}=(\sqrt{\rho_{11}},\sqrt{\rho_{22}},\dots,\sqrt{\rho_{nn}})^t
\eea
And the initial mixed state $\rho$ have the same path information with the pure state after a series of operations. So we define the path information of the mixed state $\rho$ as the path information of the corresponding pure state $\ket{\phi}$.
\bea
D_f(\rho)=D_f(\ket{\phi}\bra{\phi})=1-f(\rho_{11},\rho_{22},\dots,\rho_{nn})
\eea
It is not hard to see that this is essentially a function that only depends on the main diagonal elements.

So, we get Theorem 2 combining the above analysis.

\begin{figure}[h]
	\centering
	\includegraphics[scale=0.5]{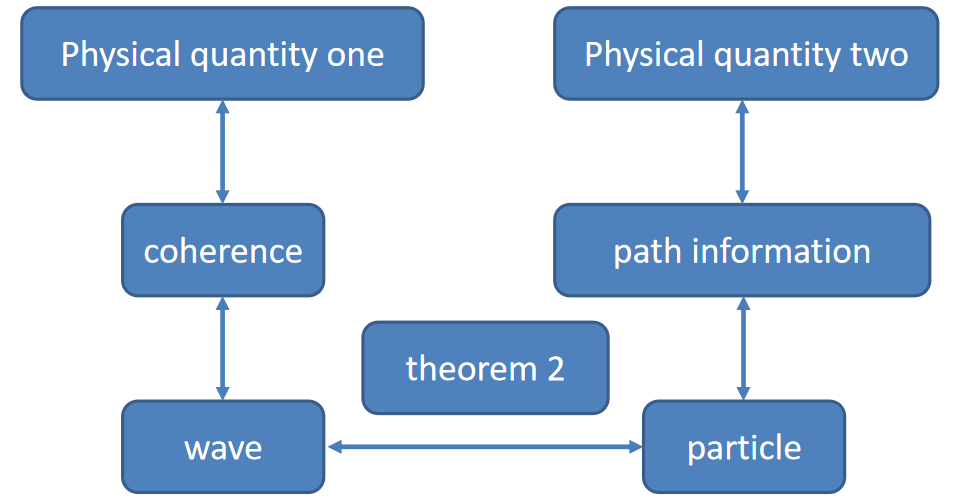}
	\caption{If a physical quantity has a functional relationship with coherence, then theorem 2 shows that the physical quantity also has a functional relationship with path information and particle, and the two functions satisfy the duality relationship, and the reverse is the same.}
	\label{FIG.2.}
\end{figure}

\begin{thrm}
	For any given coherence measure $C_f(\ket{\psi})$ for pure state $\ket{\psi}$, there will always be a corresponding path information $D_f(\rho)$ and coherence measure $C_f(\rho)$, satisfy $C_f(\rho)+D_f(\rho)\leq1$.
	
\bea
	C_f(\rho)=\min_{p_j,\rho_j}\sum_jp_jC_f(\rho_j)
\eea
\bea
	D_f(\rho)=1-f(\rho_{11},\rho_{22},\dots,\rho_{nn})
\eea
\end{thrm}

\begin{proof}
    Similar to the proof in Theorem 1.

    \hspace*{\fill}

    $(1)$ Set $\rho_{ii}=1$ for a definited $i$, then the main diagonal elements of the quantum state $\rho$ have the following form
    \bea
    (\rho_{11},\rho_{22},\dots,\rho_{nn})=(0,\dots,1,\dots,0)
    \eea
    meanwhile $D_{f}(\rho)=1-f(0,\dots,1,\dots,0)=1$, reaches its global maximum.

    $(2)$ Set the main diagonal elements of the quantum state $\rho$ as
    \bea
    (\rho_{11},\rho_{22},\dots,\rho_{nn})=(1/n,1/n,\dots,1/n)
    \eea
    meanwhile $D_{f}(\rho)=1-f(1/n,1/n,\dots,1/n)=1$, reaches its global minimum.

    $(3)$ and $(4)$ are same with the proof of theorm 1.

    \hspace*{\fill}
	
	For a mixed state $\rho$, there are many pure state decompositions, and we choose one of them at will,
	\bea
	\rho=\sum_{i}p_i\ket{\psi}_i\bra{\psi}_i
	\eea
	We can deduce that
	\bea
	\begin{aligned}
		f(\rho_{11},\rho_{22},\dots,\rho_{nn})&=f(\sum_{i=1}^{n}p_i|\psi_{11}^i|^2,\sum_{i=1}^{n}p_i|\psi_{22}^i|^2,\dots,\sum_{i=1}^{n}p_i|\psi_{nn}^i|^2)\\
		&=f(\sum_{j=1}^{n}p_j(|\psi_{11}^j|^2,|\psi_{22}^j|^2,\dots,|\psi_{nn}^j|^2)^t)\\
		&\geq\sum_{j=1}^{n}p_jf(|\psi_{11}^j|^2,|\psi_{22}^j|^2,\dots,|\psi_{nn}^j|^2)^t\\
		&\geq\min_{p_j,\rho_j}\sum_{j}p_jC_f(\rho_j)\\
		&=C_f(\rho)
	\end{aligned}
	\eea
	So, for a mixed state $\rho$, we get
	\bea
	D_f(\rho)+C_f(\rho)\leq1
	\eea
\end{proof}

Thus, we extend the complementarity in the pure state to the mixed stat. As we said earlier, for any given coherence measure, We can identify the specific form of the complementary relationship by theorem 2. Moreover, if we can measure coherence by a physical quantity, we can also measure path information by that physical quantity, conversely, we can also apply physical quantities associated with path information to measure coherence.

\section{a generalized triaility}

In the following, we analyze why $C_f(\rho)$ and $D_f(\rho)$ for mixed state are not neat complementary. Note that the following analysis only applies to the absence of detectors, as the resource theory of path distinguishability has not been established. We define
\bea
M_f(\rho)=f(\rho_{11},\rho_{22},\dots,\rho_{nn})-\min_{p_j,\rho_j}\sum_{j}p_jC_f(\rho_j)
\eea
so we can get
\bea
D_f(\rho)+C_f(\rho)+M_f(\rho)=1
\eea
This is exactly the triaility what we were hoping for.

Recently, resource theory on predictability was established, and discussed the relation among the resource theories of predictability, coherence and purity, pointed that complementarity relations can be used as purity measures \cite{ref45}. Because purity and mixdness are two opposite concepts, so, taking our defined complementarity
\bea
C_{f}(\rho)+D_{f}(\rho)=P_{f}(\rho)
\eea
as a purity monotone, the complementary $M_{f}(\rho)$ is naturally a quantity associated with mixdness.

Although $M_{f}$ is quite closely related to the mixdness, given that it is acquired from a fixed basis, we are not sure whether it is a significant measure of mixdness. This is because the mixdness of quantum states is independent of their eiqenstates, what we define as $M_{f}$ can intuit from that definition that it is related to a fixed basis. But in practical experiments and theoretical analysis, we always have to fix a set of bases in advance, so the following discussions of $M_{f}(\rho)$ are based on fixed basis. From the special form of function $f$, we can show that $M_{f}$ satisfies some properties as a degree of mixdness.

Next, we show that $M_{f}(\rho)$ satisfies several basic properties of mixdness and has a certain monotonicity.

\begin{thrm}
	$M_{f}(\rho)$ meets the following properties:
	
	\hspace*{\fill}
	
		(1)	$M_f(\rho)$ reaches its global minimum if the state $\rho=\ket{\psi}\bra{\psi}$ is a pure state.
		
		(2)	$M_f(\rho)$ reches its global maxmum if the state $\rho=\mathbbm{1}/n$ is max mixdness state.
		
		(3)	$M_f(\rho)$ satisfies concave.
		
		(4) $M_f(\rho)$ is monotonicity for quantum states that satisfy condition $\sigma=p\rho+(1-p)\mathbbm{1}/n$.
		
		(5) $M_f(\rho)$ is monotonicity for quantum states that satisfy condition $\sigma=p\rho+(1-p)\triangle(\rho)$.
		
		(6) $M_f(\rho)$ is monotonicity for two-dimensional quantum states that satisfy condition $\sigma=p\rho+(1-p)[\mathbbm{1}/2+(\rho-\triangle(\rho))]$.
\end{thrm}

The first three terms in theorem 3 are the basic properties of the degree of mixdness, and clauses (4), (5) and (6) indicate that $M_{f}(\rho)$ has a certain monotonicity.

\begin{proof}
	
	\hspace*{\fill}
		
	$(1)$ Set $\rho=\ket{\psi}\bra{\psi}$ is a pure state, at this point, quantum states do not have a degree of mixdness, meanwhile, the $M_{f}(\rho)$ we defined will not exist, i.e. $M_{f}(\ket{\psi}\bra{\psi})=0$. The triality naturally transforms into the wave-pariticle relationship.
		
	$(2)$ Set $\rho=\mathbbm{1}/n$ is the maxmally mixed state, at this point the quantum state $\rho$ has the greatest mixdness. Obviously, quantum state $\rho$ has no coherence and path information, becasuse of the maxmally mixed state has no non-primary diagonal elements and the main diagonal elements are all $1/n$, which means that $M_{f}(\rho)$ reaches its maximum $M_{f}(\rho)=1$.
		
	$(3)$ For any decomposition of state $\rho=\sum_{i}p_{i}\rho_{i}$, it satisfies concave,
	\begin{equation}
		\begin{aligned}
			M_{f}(\sum_{i}p_{i}\rho_{i})&=1-C_{f}(\sum_{i}p_{i}\rho_{i})-D_{f}(\sum_{i}p_{i}\rho_{i})\\
			&\geq1-\sum_{i}p_{i}C_{f}(\rho_{i})-\sum_{i}p_{i}D_{f}(\rho_{i})\\
			&=\sum_{i}p_{i}(1-C_{f}(\rho_{i})-D_{f}(\rho_{i}))\\
			&=\sum_{i}p_{i}M(\rho_{i})
		\end{aligned}
	\end{equation}
    beacuse of that $C_{f}$ and $D_{f}$ are convex.

	$(4)$ For any given state $\rho$, we find the states satisfies that
	\bea
	\sigma=p\rho+(1-p)\mathbbm{1}/n
	\eea
	which is the common form of mixed state. Obviously, the mixdness of $\sigma$ is higher than the mixdness of $\rho$.
	
	According to the concave of $M_{f}(\rho)$, we can get that
	\begin{equation}
		\begin{aligned}
			M_{f}(\sigma)&=M_{f}(p\rho+(1-p)\mathbbm{1}/n)\\
			&\geq pM_{f}(\rho)+(1-p)M_{f}(\mathbbm{1}/n)\\
			&\geq M_{f}(\rho)
		\end{aligned}
	\end{equation}

    $(5)$ For two given states $\rho$ and $\rho_{1}$, the following relationship is satisfied
    \bea
    \rho_{1}=p\rho+(1-p)\triangle(\rho)
    \eea
    state $\rho_{1}$ and state $\rho$ have the same principal diagonal elements, so the path information for both states is the same, meanwhile, the degree of similarity between state $\rho_{1}$ and the max mixed state is greater than that of state $\rho$, so the mixdness of state $\rho_{1}$ is greater than the mixdness of state $\rho$.

    Correspondingly, we can get
    \begin{equation}
    	\begin{aligned}
    		M_{f}(\rho_{1})&=M_{f}(p\rho+(1-p)\triangle(\rho))\\
    		&\geq pM_{f}(\rho)+(1-p)M_{f}(\triangle{\rho})\\
    		&\geq M_{f}(\rho)
    	\end{aligned}
    \end{equation}

    $(6)$ Similar to $(5)$, for two given two dimensional states $\rho$ and $\rho_{2}$, the following relationship is satisfied
    \bea
    \rho_{2}=p\rho+(1-p)[\mathbbm{1}/2+(\rho-\triangle(\rho))]
    \eea
    where $\rho^{'}=\mathbbm{1}/2+(\rho-\triangle(\rho))$ satisfies that
    \bea
    \rho^{'}_{12}=\rho_{12}=\sqrt{\rho_{11}\rho_{22}}\leq \sqrt{\rho^{'}_{11}\rho^{'}_{11}}=1/2
    \eea
    so the determinant of $\rho^{'}$ is greater than zero, which means that $\rho^{'}$ is a positive definite matrix.
    \bea
    |\rho^{'}|=|\rho_{11}^{'}\rho_{22}^{'}-(\rho_{12}^{'})^{2}|=\frac{1}{4}-(\rho_{12}^{'})^{2}\geq 0
    \eea
    Now we can say that $\rho^{'}$ is a martix of quantum states. Similar to $(5)$, state $\rho_{2}$ and state $\rho$ have the same non-primary diagonal elements, so the coherence for both states is the same, meanwhile, the mixdness of state $\rho_{2}$ is greater than the mixdness of state $\rho$.

    And we can get
    \bea
    M_{f}(\rho_{2})\geq M_{f}(\rho)
    \eea
    refer to (13).
\end{proof}

For easy of understanding, we use a more intuitive way to describe the mixdness measure $M_{f}(\rho)$ we defined and the properties expressed in Theorem 3.

\begin{figure}[h]
	\centering
	\includegraphics[scale=0.4]{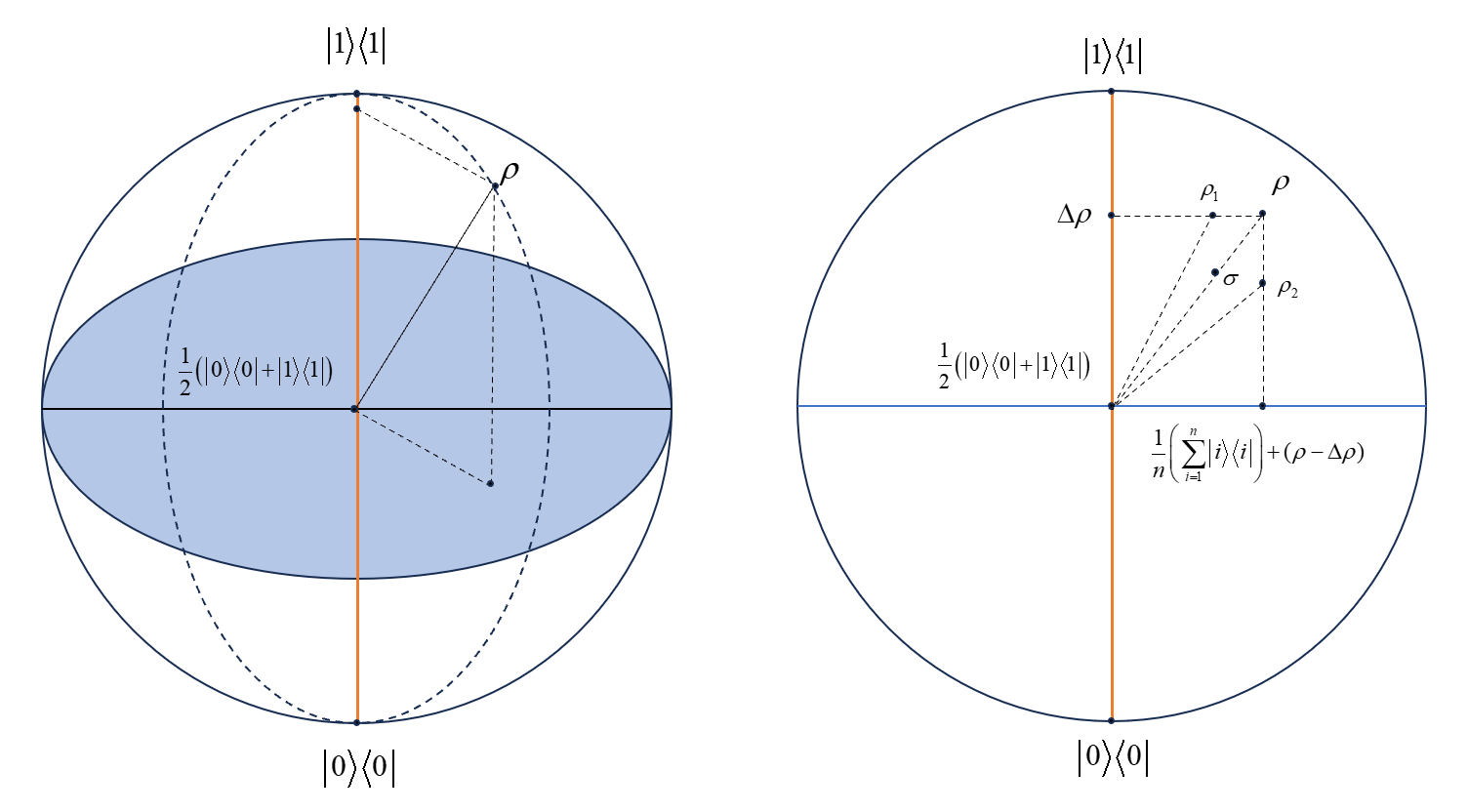}
	\caption{A diagram of Bloch's sphere, where the set of free states $F_{predictability}$ is represented as a blue plane(XY-plane), the set of free states $F_{coherence}$ is represented as a red line(Z-axes), and the set of free states $F_{purity}$ is represented as the center of the sphere.}
	\label{FIG.3.}
\end{figure}

In an intuitive way, it is well known that two dimensional states correspond one-to-one to points on the Bloch sphere, the pure state corresponds to a point on the sphere and the mixed state corresponds to a point inside the sphere, as showen in Figure.3. For a given state $\rho$, coherence measures are represented as the perpendicular line from point $\rho$ to Z-axes (free states set of coherence), and the path predictability are represented as the perpendicular line from point $\rho$ to XY-plane (free states set of path predictability) \cite{ref45}, which shows that the coherence and path predictability are only relevant for non-primary diagonal elements and primary diagonal elements. Purity is expressed as the connection of point $\rho$ to the center of the sphere (free state of purity) \cite{ref46}, and conversely, the mixdness is expressed as the shortest distance from point $\rho$ to the sphere. This is partly enough to help us understand the triaility.

For a given state $\rho$, We make three perpendicular lines to each of the three sets of free states as follow, where $F_{purity}(\rho)$ is the free states set of purity, $F_{coherence}$ is the free states set of coherence and $F_{predictability}$ is the free states set of path predictability. Since the resource theory of mixdness has not been established, becaues of the free states set of mixdness is a nonconvex set, so we analyze mixdness through the resource theory of purity here.
\bea
F_{purity}(\rho)=\{\mathbbm{1}/n\}
\eea
\bea
F_{coherence}(\rho)=\{\rho|\rho=\sum_{i}\ket{i}\bra{i}\}
\eea
\bea
F_{predictability}(\rho)=\{\rho|\rho=(1-p)\frac{1}{n}+p\ket{\psi_d}\bra{\psi_d}\}
\eea
where $\ket{\psi_{d}}=\frac{1}{\sqrt{d}}\sum_{j}e^{i\phi_{j}}\ket{x_j}$, the main  diagonal elements of states belong to set $F_{path}$ are $\frac{1}{n}$, and the non-principal diagonal elements can be any value \cite{ref45}. So the points on all three lines satisfy the monotonicity of the $M_{f}(\rho)$ from (4), (5) and (6) in Theorem 3.

For an instance, given a state $\rho$, makes a perpendicular line towards the Z-axis, the length of this line represents the coherence of this state $\rho$, and the states on this line has the same path predictability with the state $\rho$. Let's take a state $\rho_{1}$ on the line, we can get the mixdness of $\rho_{1}$ is greater than $\rho$, because of state $\rho_{1}$ is closer to the maximum mixed state than state $\rho$. And we can get that $M_{f}(\rho_{1})\geq M_{f}(\rho)$ from $(4)$ of Theorem 3, which means $M_{f}$ satisfies monotonicity for the states on the line. Similarly, monotonicity is also satisfied on the line between state $\rho$ and the center of the sphere and on the vertical line between state A and the XY-plane.

\section{Example}

Next we will give two examples to illustrate our proposed theorem.

For some well-formed coherence measures, the theorms we have established are also valid, just like the coherence measure $C_{l1}(\ket{\psi})=\frac{1}{n-1}\sum_{i\neq j}|\psi_{ij}|=\frac{1}{n-1}\sum_{i\neq j}|c_{i}||c_{j}|$ difined by L-1 norm, we also can get following functions.
\bea
C_{l1}(\ket{\rho})=\frac{1}{n-1}\sum_{i\neq j}|\rho_{ij}|
\eea
\bea
D_{l1}(\rho)=1-\frac{1}{n-1}\sum_{i\neq j}\sqrt{|\rho_{ii}||\rho_{jj}|}
\eea
\bea
M_{l1}(\rho)=\frac{1}{n-1}\sum_{i\neq j}\sqrt{|\rho_{ii}||\rho_{jj}|}-|\rho_{ij}|
\eea
This is consistent with the result $C(\rho)+D(\rho)\leq1$ in \cite{ref23}, but we have established a triaility relationship $C(\rho)+D(\rho)+M(\rho)=1$ on this basis, $D(\rho)$ is the path predictability. And we can also get $C_{l1}(\rho_{s})+D_{l1}(\rho_{s})+M_{l1}(\rho)\leq1$ for interferometer with detectors bacause of
\bea
\begin{aligned}
	M_{l1}(\rho_{s})&=\frac{1}{n-1}\sum_{i\neq j}\sqrt{|\rho_{ii}||\rho_{jj}||d_{i}|^2|d_{j}|^2}-|\rho_{ij}||\bra{d_{i}}\ket{d_{j}}|\\
	&=\frac{1}{n-1}\sum_{i\neq j}\sqrt{|\rho_{ii}||\rho_{jj}|}-|\rho_{ij}||\bra{d_{i}}\ket{d_{j}}|\\
	&\geq\frac{1}{n-1}\sum_{i\neq j}\sqrt{|\rho_{ii}||\rho_{jj}|}-|\rho_{ij}|\\
	&=M_{l1}(\rho)
\end{aligned}
\eea
$D(\rho_{s})$ is the path distinguishabililty.

So, not only do we establish a triality relationship without detectors, we also establish a triality relationship with detectors based on the special form of function $f$.

Next, we mainly introduce how to establish a triality relationship for the general form coherence measurement, that is, the coherence measurement in the mixed state through the convex roof structure. Fidelity, as an important physical quantity, plays a vital role in quantum information theory, however, as mentioned above, the coherence measure defined by fidelity does not have a good form, and the coherence measure can only be established in the mixed state by the method of convex roof construction \cite{ref43}, and the coherence measure difined by fidelity is
\bea
C_{F}(\ket{\psi})=\min_{\sigma\in{I}}\sqrt{1-F(\ket{\psi},\sigma)}
\eea
where $F(\rho,\sigma)=(Tr\sqrt{\sqrt{\rho}\sigma\sqrt{\rho}})^2$. Set $\ket{\psi}=\sum_{i}c_{i}\ket{i}$, and we can get that $C_{F}(\ket{\psi})=\sqrt{1-|c_{i}|^{2}}$, so according to theorem 1 and theorem 2, we can get the following functions.
\bea
C_{F}(\rho)=\min_{p_n,\psi_n}\sum_{n}p_{n}\sqrt{1-|c_{i}^{(n)}|^{2}}
\eea
\bea
D_{F}(\rho)=1-\sqrt{1-|\rho_{ii}|}
\eea
\bea
M_{F}(\rho)=\sqrt{1-|\rho_{ii}|}-\min_{p_n,\psi_n}\sum_{n}p_{n}\sqrt{1-|c_{i}^{(n)}|^{2}}
\eea
It is easy to verify that $D_{F}$ can be used as a predictability measure\cite{ref45}, and for interferometers with or without detectors, the wave-particle relationship $C_{F}(\rho)+D_{F}(\rho)\leq1$ can be held, the difference is only that the independent variable in the function is $\rho$ or $\rho_{s}$. Not only that, but for interferometer without detectors, we can also establish a triaility relationship $C_{F}(\rho)+D_{F}(\rho)+M_{F}(\rho)=1$, $M_{F}$ as a representation of mixdness, satisfies the properties in theorem 3. Thus, we obtain a measure of predictability and particle properties through the coherence measure defined by fidelity.


So far, we have shown by two examples that a wave-particle relationship can be established for any form of coherence measure, and $D_{f}$ established in this way can also be used as a measure of predictable resources. The third term is closely related to mixdness, and the $M_{f}$ established for some special form of coherence measure has more complete properties, which depends entirely on the form of the function $f$.

\section{Conclusion}

We end up with a generalized wave-particle-mixedness triality, which is applicable to any form of coherence measure and path information. We start with the pure state case, for any form of coherence measure, by transforming the given coherence measure into a function form related to the main diagonal elements, the corresponding path information is found in a function form, and we proved that the defined path information meets the requirement of particle feature quantification. Then we extend this form to mixed state case and establish a generalized wave-particle-mixedness triality. Although the characterization of mixedness defined by us is not a complete measure of mixedness, it also meets some basic requirements for mixdness measurement and has a certain monotonicity for some states that satisfy certain relationships. From the perspective of resource theory, the form of path information we propose can be used as a measure of predictability. The triality we have established also links between resource theories such as coherence, predictability, purity and mixdness.

At last, we give two examples of how we can use our theorem to establish a trality relation for a given coherence measure, whether the coherence measure has a good form in the mixed state or not. It is also pointed out that the mixdness defined for some special coherence measures has more complete properties.
At the same time, we look forward to apply this method to other analytical methods, and build new triality relationships, such as establish a general triality relationship for the interferometer with detectors. And we look forward to applying our theorem to other concrete coherence measures and analyze whether the mixdness term has more complete properties.

\section{Acknowledgments}

This project is supported by the National Natural Science Foundation of China (Grants No. 12050410232).

\appendix

%

\end{document}